\documentclass[10pt,letterpaper,twocolumn,twoside]{IEEEtran}

\usepackage[usenames,dvipsnames]{xcolor}
\usepackage{cite}
\usepackage{ifthen}
  \usepackage{pgfplots}
\usepackage{amsmath,amssymb}

\usepackage{footnote}

\usepackage{amsmath}
\usepackage{mathtools}
\usepackage{amssymb}
\usepackage{stmaryrd}
\usepackage{dsfont}

\usepackage{array}
\usepackage{fixltx2e}
\usepackage{stfloats}
\usepackage{bbm}
\usepackage{amsthm}
\theoremstyle{plain}
\usepackage{multirow}

\usepackage{url}

\DeclareMathAlphabet{\pala}{OT1}{pag}{m}{sl}
\DeclareMathAlphabet{\eurm}{U}{eur}{m}{n}
\DeclareMathAlphabet{\mathbsf}{OT1}{cmss}{bx}{n}
\DeclareMathAlphabet{\mathssf}{OT1}{cmss}{m}{sl}
\DeclareMathAlphabet{\mathcsf}{OT1}{cmss}{sbc}{n}

\newcommand{\randomvalue}[1]{{\uppercase{#1}}}

\DeclareSymbolFont{bsfletters}{OT1}{cmss}{bx}{n}  
\DeclareSymbolFont{ssfletters}{OT1}{cmss}{m}{n}
\DeclareMathSymbol{\bsfGamma}{0}{bsfletters}{'000}
\DeclareMathSymbol{\ssfGamma}{0}{ssfletters}{'000}
\DeclareMathSymbol{\bsfDelta}{0}{bsfletters}{'001}
\DeclareMathSymbol{\ssfDelta}{0}{ssfletters}{'001}
\DeclareMathSymbol{\bsfTheta}{0}{bsfletters}{'002}
\DeclareMathSymbol{\ssfTheta}{0}{ssfletters}{'002}
\DeclareMathSymbol{\bsfLambda}{0}{bsfletters}{'003}
\DeclareMathSymbol{\ssfLambda}{0}{ssfletters}{'003}
\DeclareMathSymbol{\bsfXi}{0}{bsfletters}{'004}
\DeclareMathSymbol{\ssfXi}{0}{ssfletters}{'004}
\DeclareMathSymbol{\bsfPi}{0}{bsfletters}{'005}
\DeclareMathSymbol{\ssfPi}{0}{ssfletters}{'005}
\DeclareMathSymbol{\bsfSigma}{0}{bsfletters}{'006}
\DeclareMathSymbol{\ssfSigma}{0}{ssfletters}{'006}
\DeclareMathSymbol{\bsfUpsilon}{0}{bsfletters}{'007}
\DeclareMathSymbol{\ssfUpsilon}{0}{ssfletters}{'007}
\DeclareMathSymbol{\bsfPhi}{0}{bsfletters}{'010}
\DeclareMathSymbol{\ssfPhi}{0}{ssfletters}{'010}
\DeclareMathSymbol{\bsfPsi}{0}{bsfletters}{'011}
\DeclareMathSymbol{\ssfPsi}{0}{ssfletters}{'011}
\DeclareMathSymbol{\bsfOmega}{0}{bsfletters}{'012}
\DeclareMathSymbol{\ssfOmega}{0}{ssfletters}{'012}

\newcommand{\rvF}{{\randomvalue{F}}}
\newcommand{\rvG}{{\randomvalue{G}}}

\newcommand{\rvK}{{\randomvalue{K}}}
\newcommand{\rvX}{{\randomvalue{X}}}  	
\newcommand{\rvZ}{{\randomvalue{Z}}}	

\newcommand{\calK}{{\mathcal{K}}}

\newcommand{\calS}{{\mathcal{S}}}

\newcommand{\calX}{{\mathcal{X}}}
\newcommand{\calY}{{\mathcal{Y}}}

\newcommand{\calZ}{{\mathcal{Z}}}

\newcommand{\avgI}[1]{{{\mathbb{I}}\left(#1\right)}}           
\newcommand{\avgH}[1]{{\mathbb{H}}\left(#1\right)}              

\newtheorem*{proposition*}{Proposition}
\newtheorem{lemma}{Lemma}
\newtheorem*{remark}{Remark}

\renewcommand{\leq}{\leqslant}
\renewcommand{\geq}{\geqslant}

\pgfplotsset{scaled x ticks=false,
tick label style={font=\rmfamily}}

\DeclareSymbolFont{numbers}{T1}{ptm}{m}{n}
\SetSymbolFont{numbers}{bold}{T1}{ptm}{bx}{n}
\DeclareMathSymbol{0}\mathalpha{numbers}{"30}
\DeclareMathSymbol{1}\mathalpha{numbers}{"31}
\DeclareMathSymbol{2}\mathalpha{numbers}{"32}
\DeclareMathSymbol{3}\mathalpha{numbers}{"33}
\DeclareMathSymbol{4}\mathalpha{numbers}{"34}
\DeclareMathSymbol{5}\mathalpha{numbers}{"35}
\DeclareMathSymbol{6}\mathalpha{numbers}{"36}
\DeclareMathSymbol{7}\mathalpha{numbers}{"37}
\DeclareMathSymbol{8}\mathalpha{numbers}{"38}
\DeclareMathSymbol{9}\mathalpha{numbers}{"39}

\begin{document}
\newboolean{double}
\setboolean{double}{true}

\title{
The Effect of Eavesdropper's Statistics in 
Experimental Wireless Secret-Key Generation}

\author{Alexandre J. Pierrot$^{\ast}$,~\IEEEmembership{Student Member,~IEEE}, R\'{e}mi A. Chou,~\IEEEmembership{Student Member,~IEEE}  and Matthieu R. Bloch,~\IEEEmembership{Member,~IEEE}
\thanks{A. J. Pierrot, R. A. Chou and M. R. Bloch are with the School~of~Electrical~and~Computer~Engineering,~Georgia~Institute~of~Technology, Atlanta,~GA~30332--0250, and GT-CNRS UMI 2958, 2 rue Marconi, 57070 Metz, France.} \thanks{E-mail: alex.pierrot@gatech.edu (\emph{corresponding}), remi.chou@gatech.edu, and matthieu.bloch@ece.gatech.edu}
\thanks{This research was supported in part by CNRS with grant PEPS PhySecNet, by NSF with grant CCF1320298, and by ANR with grant 13-BS03-0008.}
\thanks{Digital Object Identifier N/A}
}

\ifthenelse{\boolean{double}}{\markboth{Submitted to the IEEE Transactions on Information Forensics and Security}{Pierrot \MakeLowercase{\textit{et al.}}: Practical Limitations of Secret-Key Generation in Narrowband Wireless Environments}}{}

\maketitle
%
%
%


\begin{abstract}
This paper investigates the role of the eavesdropper's statistics in the implementation of a practical secret-key generation system. We carefully conduct the information-theoretic analysis of a secret-key generation system from wireless channel gains measured with software-defined radios. In particular, we show that it is inaccurate to assume that the eavesdropper gets no information because of decorrelation with distance. We also provide a bound for the achievable secret-key rate in the finite key-length regime that takes into account the presence of correlated eavesdropper's observations. We evaluate this bound with our experimental gain measurements to show that operating with a finite number of samples incurs a loss in secret-key rate on the order of 20\%.
\end{abstract}


\ifthenelse{\boolean{double}}{}{
\begin{center}
\bfseries EDICS Category: INF-SECC
\end{center}
}

\section{Introduction}
\label{sec:intro}

Wireless communications are, by nature, particularly vulnerable to eavesdropping and call for carefully designed encryption mechanisms. So far, the protection of wireless communications has mainly involved mathematical cryptographic schemes implemented at the upper layers of the communication stack. Such systems require a \emph{secret key} (either private and/or public) to encrypt the data stream, in such a way that it becomes extremely difficult to decipher the messages without the key. Mathematical cryptography relies on mathematical problems that are supposedly difficult to solve, such as prime number factorization or discrete logarithm computation. Although these techniques are widespread and well-tested, some of their current limitations, such as the need for complex key management, encourage the development of other mechanisms. 

Physical-layer security promises      ways of providing secrecy through the use of the intrinsic randomness present in any communication medium, such as noise, fading, and interferences. Physical-layer security techniques can be broadly divided into two categories. The first set of techniques treats the communication medium as a \emph{noisy channel}, with the goal of \emph{communicating} messages reliably between legitimate users and securely with respect to eavesdroppers. Such techniques are typically developed from an information-theoretic wiretap channel model~\cite{Wyner1975} and, when used in wireless channels, now often involve the introduction of artificial noise through cooperative jamming~\cite{Tekin2008,Tekin2008erratum}. The second set of techniques views the communication medium as a \emph{noisy source}, and aims at \emph{extracting} secret keys from the channel randomness. Such techniques are usually designed using an information-theoretic secret-key agreement model~\cite{Ahlswede1993,Maurer93}, which only focuses on generation of secret keys without independently of how they are subsequently used. In this paper we focus on the experimental implementation of the latter set of physical-layer security techniques. %


As summarized in Table~\ref{tab}, several works have already experimentally investigated the generation of secret keys from wireless channels. In fact, the gains of wireless channels provide a natural source of randomness, for which \emph{reciprocity} guarantees that legitimate users obtain strongly correlated channel observations, while \emph{diversity} ensures that the observations of a third-party eavesdropper disclose little information about the legitimate users' measurements. However, while these works are often motivated by an information-theoretic formulation, the information-theoretic modeling is not fully developed. In particular, we argue that often made assumptions about the eavesdropper's statistics, such as the decorrelation of channel gains at distances larger than half the wavelength, as well as the use of asymptotic values of secret-key generation, lead to over-simplifications of the protocols and over-estimations of the achievable information-theoretic secret-key rates. Consequently, while canonical theoretical models of wireless channels have proved incredibly useful to design reliable communication systems, their use for the design of secret-key generation systems requires more care. Similarly, checking that generated keys pass statistical tests~\cite{Jana2009,Zhang2010,Premnath2012,Ren2011}, which have been primarily designed for mathematical cryptography and only verify some desirable statistical property of the key, does not guarantee information-theoretic secrecy. 

\begin{savenotes}
\begin{table*}[t]
\caption{Review and comparison of existing literature.}
\renewcommand\arraystretch{1.4}
\begin{center}
\begin{tabular}{|c||c|p{2.8in}|p{2.2in}|} 
\hline
\textbf{Reference} & \textbf{Experiments}     & \textbf{Measurement statistics}  & \textbf{Security Analysis} \\ \hline \hline
\cite{Wang2011,Wang2012,Zhang2010,Ren2011} & {No} & \multirow{3}{*}{Postulated from model or non-applicable} &  \multirow{4}{*}{{Non information-theoretic} }\\\cline{2-2}
\cite{Chan2011,Wallace2009,Madiseh2008}  &{Simulations}&   &  \\ \cline{2-2}
\cite{Jana2009,Premnath2012,Patwari2010,Aono2005} & \multirow{3}{*}{\textbf{Yes}} &    &  \\ \cline{3-3}
\cite{Li2006} &  &{$I(X;Z)=0$ because of {decorrelation}} & \\ \cline{3-4}
\cite{Wallace2010,Ye2010,Imai2006b,Wilson2007} &   & \textbf{Estimated from experimental measurements}& {Asymptotic information-theoretic}  \\ \hline \hline
Present paper &  \textbf {Yes} & \textbf {Estimated from experimental measurements} & \textbf {Finite Length} \\ \hline
\end{tabular}
\end{center}
\label{tab}
\end{table*}
\end{savenotes}

The objective of this paper is to investigate the practical effect of eavesdropper's statistics by implementing a secret-key generation system from wireless channel gain with software-defined radios, and by carrying out a careful information-theoretic analysis. We emphasize that the weakness of previously reported system lies in the modeling of the source of randomness, but not in the operation of the subsequent protocol; therefore, we do not attempt to develop a complete secret-key generation system and we restrict our experimental system to the acquisition and processing of randomness for the channel. As reported in the paper, our analysis allows us to conclude that (1) assuming that the eavesdropper does not get any information because of decorrelation with distance is not exact in a real wireless system and (2)  the existence of a correlated eavesdropper's observation makes the evaluation of the finite-length secret-key rates much more intricate. To the best of our knowledge, and as summarized in Table~\ref{tab}, no previous experimental work has precisely looked into the effect of eavesdropper’s correlated measurements and their effect on achievable key rates in the finite length regime.


This paper is organized as follows. In Section~\ref{sec:goals}, we recall the basic principles, mathematical formulation, and common assumptions of secret-key generation from wireless channel gains. In Section~\ref{sec:esi}, we describe the characteristics of our experimental setup. In Section~\ref{sec:infl}, we assess the robustness of secret-key generation with respect to the diversity assumption.  In Section~\ref{sec:finlen}, we develop a theoretical achievable secret-key rate with a finite number of samples, which we evaluate with our experimental measurements. Finally, in Section~\ref{sec:concl-disc}, we present some concluding remarks and we discuss the possible limitations of the approach.

\section{Secret-Key Generation from Wireless Channel Gains}
\label{sec:goals}		

\subsection{Secret-key generation strategy}




The impulse response $h(t)$ of a wireless channel between two terminals results from the reflections and attenuations underwent by the transmitted signal along different paths. We focus on narrowband channels with approximately 1\,MHz bandwidth, for which the received signal is essentially a delayed version of the original one attenuated by a random complex gain $G(t)\exp(\jmath \phi(t))$. This complex gain accounts for the aggregated effect of attenuation and phase change of each individual path; note that reciprocity guarantees that the gain $G_{AB}(t)$ between two points A and B is the same as the gain $G_{BA}(t)$ between the points B and A. The coherence time during which the gain $G(t)$ and the phase $\phi(t)$ remain constant scales approximately as
$ T_c\approx{\lambda}/{v}$, where $\lambda$ is the wavelength and $v$ is the characteristic speed of the environment. For instance, in our experimental setup, the objects around the receiver move at about one meter per second, so that the coherence time is on the order of milliseconds for wireless communications at 2.5\,GHz.  For simplicity, we focus on the randomness of the channel gain $G(t)$, since precise measurements of the phase $\phi(t)$ require a synchronization of the terminals beyond the capability of our hardware.

\begin{figure}
\centering
\ifthenelse{\boolean{double}}{\includegraphics[width=8.8cm]{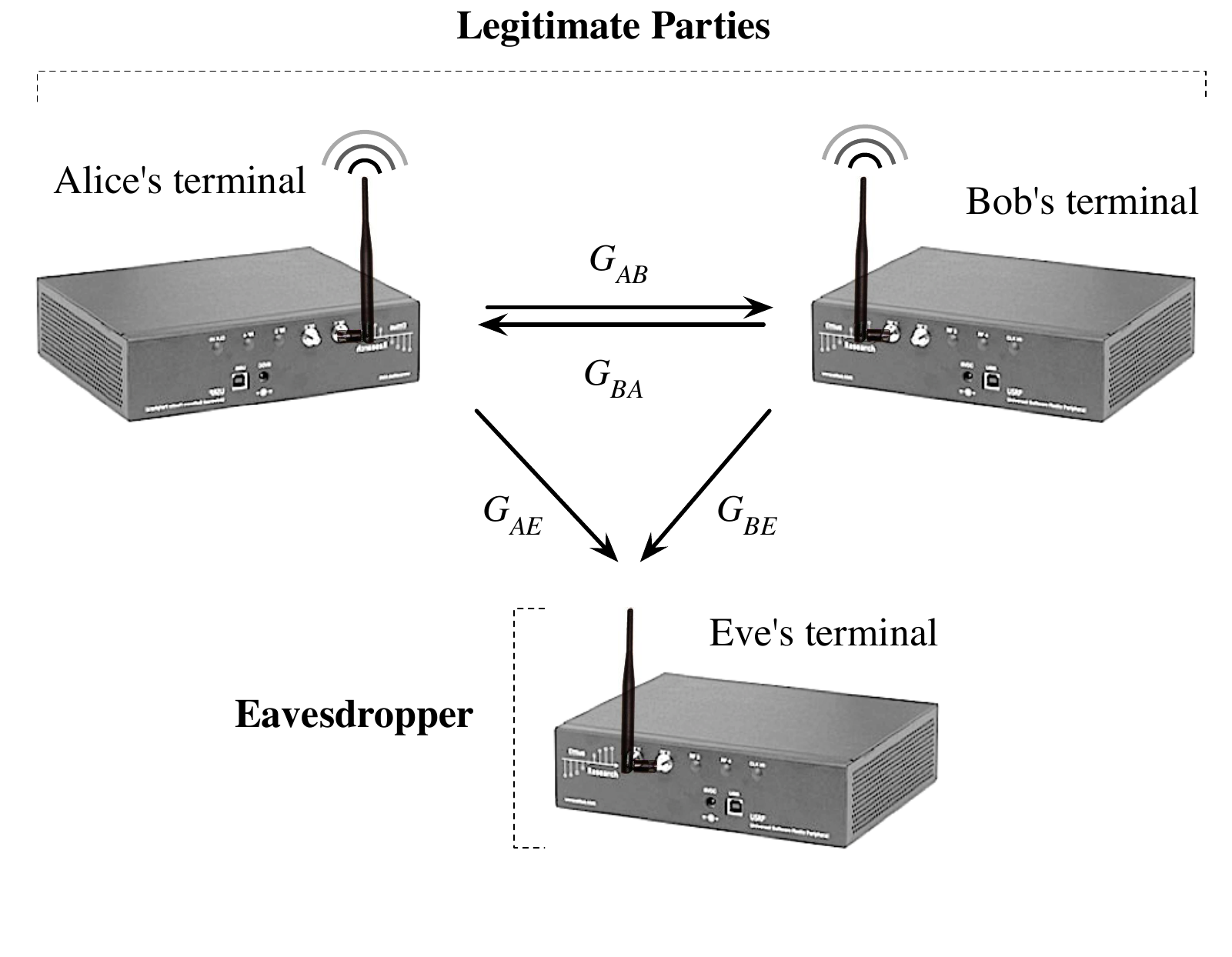}}{\includegraphics[width=10.8cm]{setup.pdf}}
\caption{Experimental measurements of wireless channel gains with software-defined radios.}
\label{fig:setup}
\end{figure} 
Following common practice, and as illustrated in Fig.~\ref{fig:setup}, we measure the channel gains between two legitimate terminals in Time-Division-Duplex mode as follows. We assume that the first terminal, Alice, sends a complex probe signal ${b}(t)$ with unit energy, whose duration $\beta$ is much smaller than the coherence time $T_c$, so that the channel gain $G$ remains almost constant over the pulse duration. The second terminal, Bob, measures from the channel a delayed and faded version $\tilde{{b}}(t)$ of ${b}(t)$. Using his knowledge of the probe signal, Bobs matches $\vec{{b}}(t)=\tilde{{b}}(t+T_{d})$  and $b(t)$  as 
\begin{equation}
T_{d}=\text{argmax}_{\tau} (\vec{{b}}\ast {b})(\tau).
\end{equation}
Then, Bob estimates the channel gain from Alice to Bob as 
\begin{equation}
G_{AB} 
=\displaystyle\int_{\beta}\left|\vec{{b}}(t)\right|^2\text{d}t,
\end{equation}
by measuring  the energy of the delayed received probe signal. Simultaneously, an eavesdropper, Eve, obtains a channel gain $G_{AE}$. Alice estimates the channel gain $G_{BA}$ in a similar fashion from Bob's probe signal, which also provides Eve with a channel gain $G_{BE}$ in the process. 





Once $n$ measurements are performed, Alice, Bob, and Eve, effectively observe the components of a noisy source $(\calX^n\calY^n\calZ^n,p_{X^nY^nZ^n})$, in which $X^n$ consists of $n$ channel gains $G_{BA}^n$, $Y^n$ consists of the channel gain $G_{AB}^n$, and $Z^n$ consists of both sequences $G_{AE}^n$ and  $G_{BE}^n$. 
A secret-key generation strategy $\calS_n$ for the source  $(\calX^n\calY^n\calZ^n,p_{X^nY^nZ^n})$ with unlimited public communications consists of the following operations:
\begin{itemize}
\item \emph{Reconciliation:} Alice transmits a public message $\rvF$ over the public authenticated channel, which allows Bob to construct an estimate $\hat{X}^n$ of $X^n$from $Y^n$ and $F$;
\item \emph{Privacy amplification:} Alice chooses a function $G$ uniformly at random in a family of universal$_2$ hash functions, which is disclosed to all parties. Alice then computes $\rvG(\rvX^n)\in \calK $ while Bob computes $G(\hat{X}^n)$. Setting $\calK\triangleq \llbracket 1,2^{nR} \rrbracket$, $R$ is called the \emph{secret-key rate}. 	
\end{itemize}    
In principle, Alice and Bob could interactively exchange messages, but we restrict the strategy to unidirectional operation. The secret-key generation strategy $\calS_n$ must ensure the following:
\begin{enumerate}
\item \emph{reliability}, measured with the probability of disagreement  $$\mathbf{P}_{\! \mathrm d}(\calS_n)\triangleq\mathbb{P}(\rvK\neq \hat \rvK|\calS_n);$$
\item \emph{(strong) secrecy}, measured by the \emph{leakage} $$\mathbf{L}(\calS_n)\triangleq \mathbb{I}(\rvK;\rvZ^n\rvF|\calS_n);$$ 
\item \emph{(strong) uniformity}, measured by 
 $$\mathbf{U}(\calS_n)\triangleq \log \left\lceil 2^{nR} \right\rceil - \mathbb{H}(\rvK|\calS_n).$$	 
\end{enumerate}
Computing the aforementioned metrics requires the knowledge of source statistics, including that of the eavesdropper.

A secret-key rate $R\triangleq\tfrac{1}{n}\log|\calK|$ is achievable if the three above metrics tend to zero as $n$ goes to infinity, and the supremum of achievable secret-key rates is called the \emph{secret-key capacity} $C_s$. Most recent works have focused on the calculation of $C_s$, which is an asymptotic limit obtained for infinitely many realizations of the source; in contrast, the analysis conducted in Section~\ref{sec:finlen} focuses on a finite length behavior that only requires  $\mathbf{L}(\calS_n)$ and $\mathbf{U}(\calS_n)$ to be small but non-zero.

\subsection{Assumptions behind the secret-key generation model}
The secrecy guaranteed by a secret-key generation strategy rely on three common assumptions. 

\paragraph{Availability of an authenticated public channel of unlimited capacity} This assumption is not unreasonable if we aim at generating low secret-key rates for which the amount of public communication is negligible compared to the channel capacity. If one explicitly introduces a  rate limitation, reconciliation with vector quantization can be used~\cite{Csiszar2000,Chou12b} without fundamentally affecting the operation of the secret-key agreement strategy.

\paragraph{Existence of enough randomness} Mobility in the environment is required to ensure that wireless channel gains have enough entropy. Mobility results from the movements of objects around the terminals or the terminals themselves; in indoor wireless environments, this channel gains experience variability as soon as people move around the communication terminals. 

 \paragraph{Knowledge of eavesdropper's statistics} In an information-theoretic secret-key generation model, one requires the knowledge of the statistical dependencies between Eve's observations and the legitimate users' to assess the secrecy of the keys. Unfortunately, there exists no indirect way to estimate these statistical dependencies of the eavesdropper without performing measurements at Eve's terminal position. In addition, as pointed out in~\cite{Jana2009}, the statistics should not be influenced by the eavesdropper to prevent the induction of artificial deterministic and predictable variations of the channel parameters. 
 
The knowledge of the eavesdropper's statistics is the most crucial assumption for the proper operation of a secret-key generation system. This could be avoided by operating in a quantum setting, e.g.~\cite{Weedbrook2012}, but to the best of our knowledge such systems are only efficiently implemented in optics. In the classical wireless setting, the assumption is often circumvented by assuming that exists enough \emph{diversity} in the environment, so that one can either assume that $\avgI{G_{AB};G_{AE}G_{BE}}=0$ meaning the eavesdropper's observations are completely independent of the legitimate users', or, at least, that $\avgI{G_{AB};G_{AE}G_{BE}}$ is upper bounded. However, we argue that this must be done with great care, and that it is crucial to precisely assess under which conditions the diversity assumption may hold, so as to define situations in which secret keys can be safely generated. Moreover, assuming that $\avgI{G_{AB};G_{AE}G_{BE}}=0$ and that the eavesdropper only observes public communication makes it considerably easier to analyze secrecy. Privacy amplification and reconciliation are simply linked using the result of Cachin and Maurer~\cite{Cachin1997}, and counting the number of bits disclosed during privacy amplification is sufficient to establish the final key length. In contrast, when $\avgI{G_{AB};G_{AE}G_{BE}}\neq 0$, the final secret key length depends on the eavesdropper's statistics and one must factor in the effect of statistical deviations from the mean when using a finite number of samples $n$. 

\section{Experimental Setup and Measurements}
\label{sec:esi}

In this section, we describe our experimental setup and our procedure to characterize the statistics of the wireless channel gains. 

\subsection{Experimental setup}

\begin{figure}
\centering
\ifthenelse{\boolean{double}}{\includegraphics[width=8.8cm]{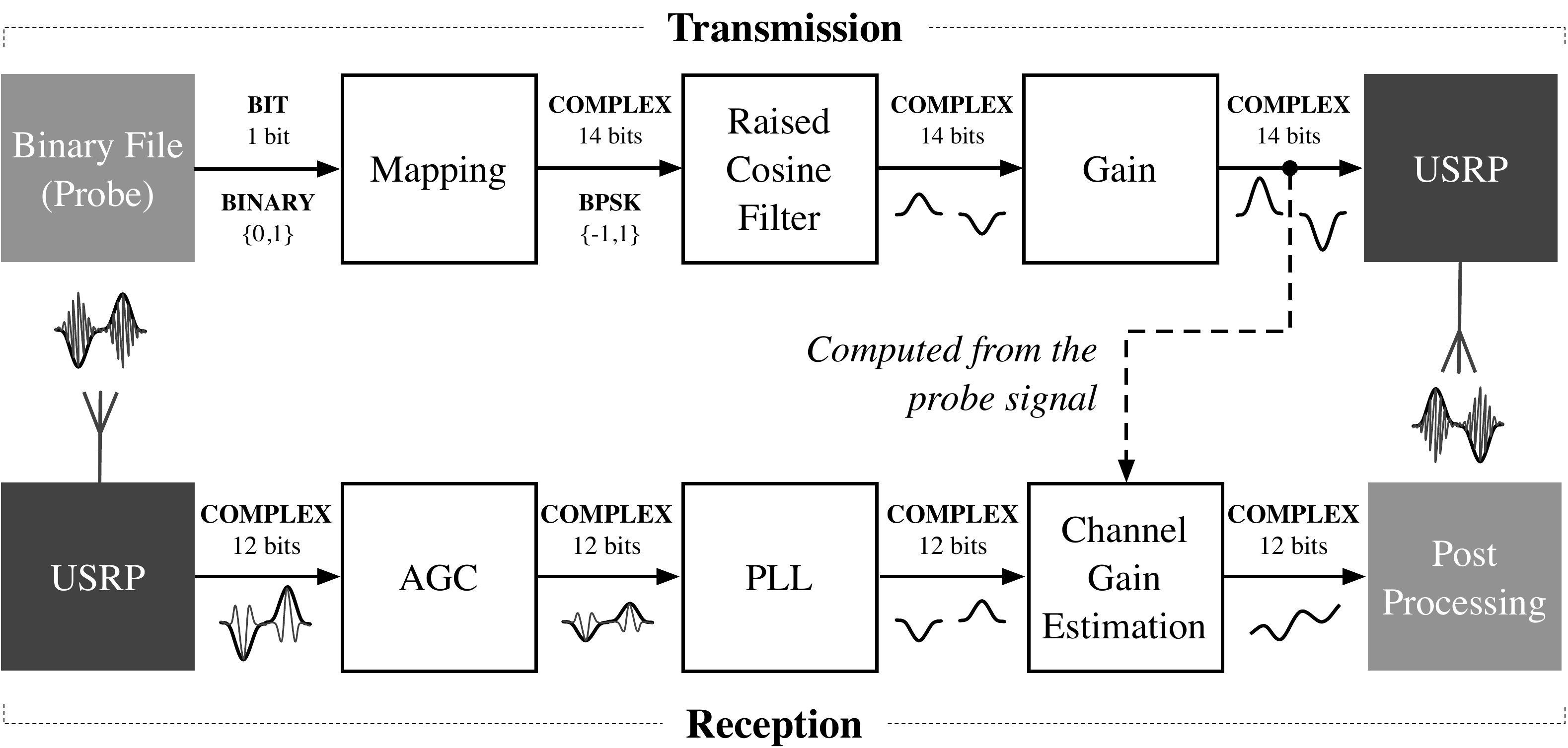}}{\includegraphics[width=13.8cm]{chain.pdf}}
\caption{Communication chain for channel gain estimation}
\label{fig:chain}
\end{figure}

\begin{figure*}
\centering
\newlength\figureheight 
    \newlength\figurewidth 
    \ifthenelse{\boolean{double}}{
    \setlength\figureheight{5cm} 
    \setlength\figurewidth{14cm}}{
    \setlength\figureheight{6cm} 
    \setlength\figurewidth{13cm}} 
    \input{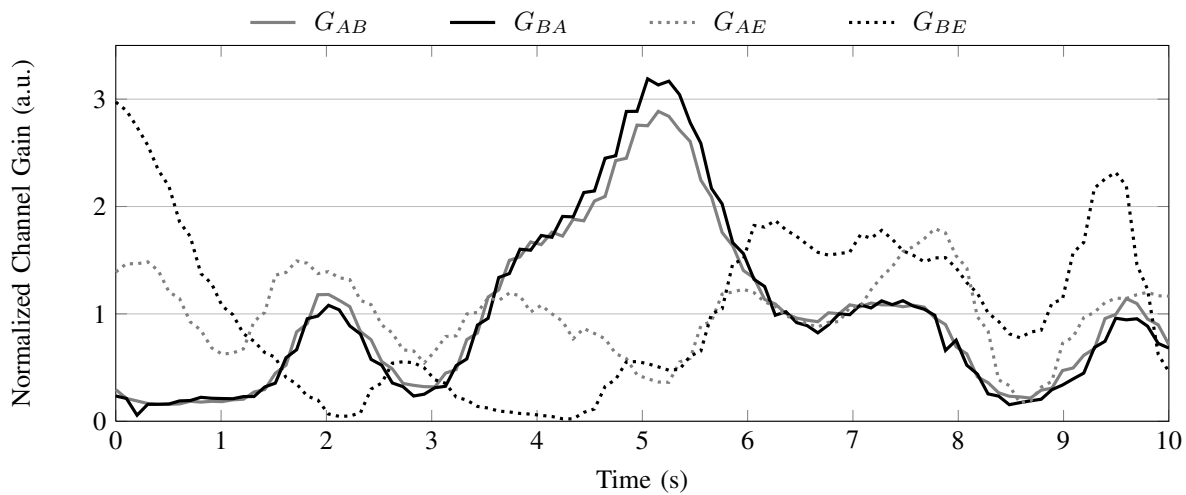}
\caption{Channel gain measurements}
\label{fig:gain}
\end{figure*}

The experiments are conducted using the first generation of USRP designed by Ettus Research\texttrademark{}. 
We use XCVR2450 daughterboards that operate in the 2.5\,GHz and 5\,GHz bands, typically used for WiFi communications. The bandwidth of the RF front-end is limited to 8\,MHz, which is well below the 20\,MHz of bandwidth required by IEEE standards. Consequently, the key rates reported in the remaining of the paper are likely to be smaller than what could be obtained on top of an actual IEEE802.11 transmission; however, this does not affect our methodology and conclusions. RF signals are transmitted using standard WiFi antennas with a transmission power below 100\,mW. The experiments are conducted in two ordinary office rooms representative of an indoor environment: one is our wireless communication laboratory and the other is a conference room. The choice of which room we used in our various experiments was only motivated by convenience. 

For convenience, all software-defined radios have the same configuration, both in hardware and software, and all are connected to a single computer that processes the transmitted and received data streams. We allow ourselves this simplification because our focus is only to study the effect of eavesdropper's statistics; while this is not exactly how a true secret-key agreement scheme would operate, we note that synchronization could be achieved in a distributed manner thanks to the pilot symbols used for channel gain measurements. The first samples of every data stream are used to calibrate the measurements and compensate hardware discrepancies by scaling all measurements to obtain the same average received energy; after calibration, the scaling is kept constant through each experiment, as we did not observe any significant drift during acquisition.

We conducted all experiments using the three-user setup represented in Fig.~\ref{fig:setup}. We used a modulation frequency of  2.484\,GHz, which corresponds to an unused WLAN channel to avoid interference with other WiFi channels. As the vast majority of communication systems~\cite{Bharadia2013}, the XCVR2450 daughterboards are limited to half-duplex operation and we were unable to simultaneously measure all the channel gains. We circumvented the problem by continually commuting the radios between the Rx and Tx modes but, because of further hardware restrictions, we could not reliably use commutation times shorter than 80\,ms. Consequently, we moved objects in the environment at less than 1\,$\text{m}.\text{s}^{\text{-1}}$ so that the channel gain would not vary much between an Rx/Tx switch, hence maintaining channel reciprocity. If the hardware allowed faster commutation to capture  faster fades, higher  secret-key generation rates would be achieved in a high mobility environment, but the security analysis would remain essentially the same.

The estimation of the channel gains is performed using a probe message sent through the communication chain described in Fig.~\ref{fig:chain}. The gain present in the transmission chain allows power control and is kept constant throughout the entire duration of the experiment. 
During reception, the USRP performs demodulation and analog-to-digital conversion. An AGC (automatic gain controller) scales the received signal to match the optimal range of the subsequent processing block. Note that it is tuned to be slow enough not to remove the gain variations over the timescale of interest. Because the system operates at a high carrier frequency, we use a phase-locked loop (PLL) to suppress any residual modulation resulting from minor differences between modulation and demodulation frequencies. The demodulated signal is then used to compute the transmission gain. Note that the probe signal is known to all parties and that the transmission chain behavior is entirely deterministic, so that all users also know the shaped signal and can compute the channel gain. The probe signal, which is a fixed randomly-generated sequence, is also used to synchronize the different radios in software.

\subsection{Characterization of channel gain statistics}
\label{ssec:char}

We now describe our methodology to characterize the statistics of the wireless channel gains for secret-key generation. The results we report next have been obtained for a fixed configuration of the terminals similar to that illustrated in Fig.~\ref{fig:setup}, in which Alice and Bob's terminals were separated by 1.5m and Eve's terminal was approximately 1m away from both Alice and Bob, and based on 500 gain measurement experiments, each lasting approximately ten seconds. Fig.~\ref{fig:gain} shows a snapshot of the evolution of the various channels gains between Alice, Bob, and Eve. As could have been expected from reciprocity, $G_{{AB}}$ closely follows $G_{{BA}}$. Eve's channels gains $G_{{AE}}$ and $G_{{BE}}$ are seemingly unrelated to the channel gains  $G_{{AB}}$ and $G_{{BA}}$, potentially confirming the existence of enough channel diversity. According to Jake's model~\cite{}, diversity should hold as soon as Eve is farther from Alice and Bob than the coherence distance, which is $\ell_{c}=\lambda/2\approx$ 6 cm, at 2.484\,GHz; in the next section, we perform a more careful diversity analysis and show that this us unfortunately not accurate enough for secret-key generation.

\begin{figure}
\centering 
    \ifthenelse{\boolean{double}}{
    \setlength\figureheight{5cm} 
    \setlength\figurewidth{7cm}}{
    \setlength\figureheight{6cm} 
    \setlength\figurewidth{9cm}} 
%
%
\begin{tikzpicture}

\begin{axis}[%
width=\figurewidth,
height=\figureheight,
scale only axis,
xmin=0,
xmax=100,
xlabel={$\text{Sample delay }\nu$},
xmajorgrids,
xtick={0,20,40,60,80,100},
ymin=0,
ymax=1.6,
ylabel={$\mathbb{I}(X_0;X_{\nu})$ (bits)},
ytick={0, 0.2, 0.4, 0.6, 0.8, 1.0, 1.2, 1.4, 1.6},
ymajorgrids
]
\addplot [
color=black,
very thick,
solid,
forget plot
]
table[row sep=crcr]{
1 1.63172881321691\\
2 1.0827348572905\\
3 0.822369537758234\\
4 0.626317374916059\\
5 0.478748446615159\\
6 0.359735048237404\\
7 0.266894301549475\\
8 0.195733992897958\\
9 0.141427558816478\\
10 0.1001798439569\\
11 0.0692335375378717\\
12 0.0465862819388832\\
13 0.0306955307016419\\
14 0.0194355193065907\\
15 0.0117040023310986\\
16 0.0067538068478933\\
17 0.00409410903744867\\
18 0.00344721191200222\\
19 0.00433592042407289\\
20 0.0067179681266338\\
21 0.0101815777942977\\
22 0.0148266531856146\\
23 0.0207312784996823\\
24 0.0274958768560659\\
25 0.0349161450317144\\
26 0.0432604337867766\\
27 0.0531658618377308\\
28 0.0633411063991534\\
29 0.0738757437484679\\
30 0.0837341243239837\\
31 0.0939275697411267\\
32 0.102293880049095\\
33 0.110035831633928\\
34 0.115462677830666\\
35 0.119549058254989\\
36 0.119420172167236\\
37 0.118475549595466\\
38 0.11472513182893\\
39 0.110089407846054\\
40 0.103413706272344\\
41 0.0970590515987108\\
42 0.0891959553943001\\
43 0.0822306600930748\\
44 0.0737211117178161\\
45 0.0663368237994982\\
46 0.059194977681853\\
47 0.0523747150913003\\
48 0.0453421760731022\\
49 0.0395463063251503\\
50 0.0362314806007969\\
51 0.0331296800102452\\
52 0.0305741897168852\\
53 0.0291627701876827\\
54 0.0281094456519083\\
55 0.0274718209250046\\
56 0.0267417738114142\\
57 0.0254916800935418\\
58 0.0241026883796282\\
59 0.0233269677820602\\
60 0.0223476022663265\\
61 0.0215193233079907\\
62 0.020099249225061\\
63 0.018971702569262\\
64 0.0168504135491057\\
65 0.0154946614886491\\
66 0.0134572190026545\\
67 0.0119842895013129\\
68 0.0108574159744547\\
69 0.0110476077406612\\
70 0.0123433080462884\\
71 0.0149675999126169\\
72 0.0175649112060732\\
73 0.0214023119927569\\
74 0.0248795996868594\\
75 0.0282714245502649\\
76 0.0303150971072419\\
77 0.031867266044749\\
78 0.0306425429249995\\
79 0.0287934158949503\\
80 0.0247367105999181\\
81 0.0216205536798515\\
82 0.0195255565862384\\
83 0.0181523207280532\\
84 0.0171274872709694\\
85 0.0171774480405749\\
86 0.0186055602546669\\
87 0.0194636533448692\\
88 0.019992411376751\\
89 0.0213963178194406\\
90 0.022073358246323\\
91 0.0225495841878028\\
92 0.0240017775857935\\
93 0.0257085308942519\\
94 0.026789616230789\\
95 0.0291403621505835\\
96 0.029919451981617\\
97 0.0305615372026933\\
98 0.0311532347338444\\
99 0.0378852977002976\\
100 0.0447675657876099\\
};
\end{axis}
\end{tikzpicture}%
\caption{Evolution of the statistical dependence between channel gains}
\label{fig:corr}
\end{figure}
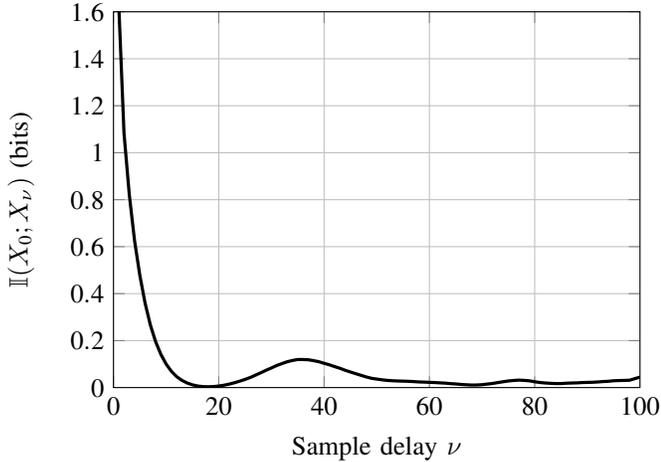

To make the statistical characterization tractable, it is desirable to operate on a \emph{memoryless} source of randomness for which two consecutive observations are independent. We thus need to downsample the raw measurements in Fig.~\ref{fig:gain} to remove the time correlation and only keep a single sample per coherence interval; since the coherence time is on the order of magnitude of  $ {\lambda}/{v}$, we expect $T_c$ to be on the order of one second. We obtain a more precise characterization of the value of $T_c$ with an estimation of the mutual information $\mathbb{I}(\rvX_0;\rvX_{\nu})$ between a sample $\rvX_0$ and the \mbox{$\nu$-th} next sample $\rvX_{\nu}$, obtained by viewing each of our experimental time series as the realization of the same ergodic random process. The lower $\mathbb{I}(\rvX_0;\rvX_{\nu})$ is, the less dependent the samples are. To use more samples for the estimation, we also assumed that the gains were wide sense stationary, which we confirmed by verifying that the quantity $\mathbb{I}(\rvX_0;\rvX_{\nu})$ remained the same for different choices of $\rvX_0$. Unless mentioned otherwise, all information metrics are estimated with the technique presented in~\cite{Peng05}. As shown in Fig.~\ref{fig:corr}, the mutual information $\mathbb{I}(\rvX_0;\rvX_{\nu})$ decays rapidly and vanishes after a dozen samples, corresponding to approximately one second, as expected. We note that operating on the down-sampled measurements instead of the original measurements would result in a lower achievable secret-key rate, which might seem an unnecessary simplification since we know how to characterize achievable secret-key rates for sources with memory~\cite{Chan2011,Bloch2011e}. However, without an accurate parametric model, the estimation of the statistics of a source with memory turns out to be a much more difficult problem.


The final step is then to estimate the joint statistics $p_{XYZ}$ of the memoryless source, which poses two challenges. First, one would in principle need to analyze the estimation error and include it in the subsequent calculation of achievable key rates; in this paper, we do not take this into account and assume that our estimation is accurate enough to be used as the true joint statistics. Second, our measurements only provide us access to \emph{quantized} measurements $X_Q$, $Y_Q$ and $Z_{Q}$ of the true channel gains $X$, $Y$, and $Z$, respectively. The quantization of $X$ and $Y$ is not critical, since the effect is merely to reduce achievable secret-key rates and to affect the reconciliation step. However, the quantization of $Z$ results in an underestimation of the eavesdropper's knowledge, and we need to assume that the eavesdropper is able to keep $Z$ continuous. Consequently, we need to estimate the joint statistics $p_{X_QY_QZ}$. Since we cannot acquire a continuous-valued $Z$ with the software-defined radios, we first construct a histogram corresponding to $p_{{X}_Q{Y}_Q{Z}_{Q'}}$ from the measured data, where $Z_{Q'}$ is a quantized version of $Z$. To obtain $p_{X_QY_QZ}$, we then perform an interpolation of the histogram with respect to $Z_Q'$. The raw data was acquired with a 14 bits resolution, which we further quantized to obtain a $4$-bit resolution for $X_Q$ and $Y_Q$, and a $6.5$-bit resolution for $Z_{Q'}$. The estimation process could be further refined, but is left for future research. 

\section{Robustness of the Diversity Assumption}
\label{sec:infl}
%




To verify to what extent the diversity assumption holds in a narrowband wireless setting, we conducted a series of measurements in our building conference room. The room is about 40 squared meters and is furnished as shown in Fig.~\ref{fig:nomotion}. Experiments were conducted off-hours to avoid any unwanted motion outside of the room. Two software-defined radios were placed in the middle of the room on the conference table, two meters apart.  We  used a third radio to represent the eavesdropper, which was then moved in the room across 60 possible positions. We measured the channel gains $G_{AB}$ and $G_{AE}$ obtained  by Bob and Eve to evaluate $\avgI{G_{AB},G_{BE}}$. These experiments only involved one way communications (Alice-to-Bob and Alice-to-Eve), thus avoiding the problem of half-duplex operation and allowing us to gather data at a faster pace. Each experiment lasted one minute, during which we acquired 50,000 channels gain values at a 1\,kHz sampling rate. 

The results of the measurement campaign are presented in Figures~\ref{fig:nomotion} and~\ref{fig:withmotion}. Eve is placed across the positions indicated by the black "+" marks, which correspond to a coarse square grid of one meter and additional positions to cover interesting spots and the room borders.  The brightness represents the \emph{normalized secrecy-rate} (supposing $G_{AB}=G_{BA}$) between the gains obtained by Eve and those obtained by Bob, which is computed as 
\begin{align}
\text{Normalized secrecy-rate} &\triangleq \dfrac{\avgI{G_{AB};G_{BA}}-\avgI{G_{AB};G_{AE}}}{\avgI{G_{AB};G_{BA}}} \notag\\
&=1-\dfrac{\avgI{G_{AB};G_{AE}}}{\avgH{G_{AB}}},
\end{align} 
where $G_{AB}$ and $G_{AE}$ are  the channel gains measured by Bob and Eve, respectively. We introduce this normalization to compensate the entropy variations of the wireless channel gains  across different experiments. This quantity is close to one (white) when the gains are independent, and equal to zero (black) when there is a one-to one mapping between $G_{AB}$ and $G_{AE}$.

\begin{figure}
\centering
    \ifthenelse{\boolean{double}}{
\includegraphics[width=8.5cm]{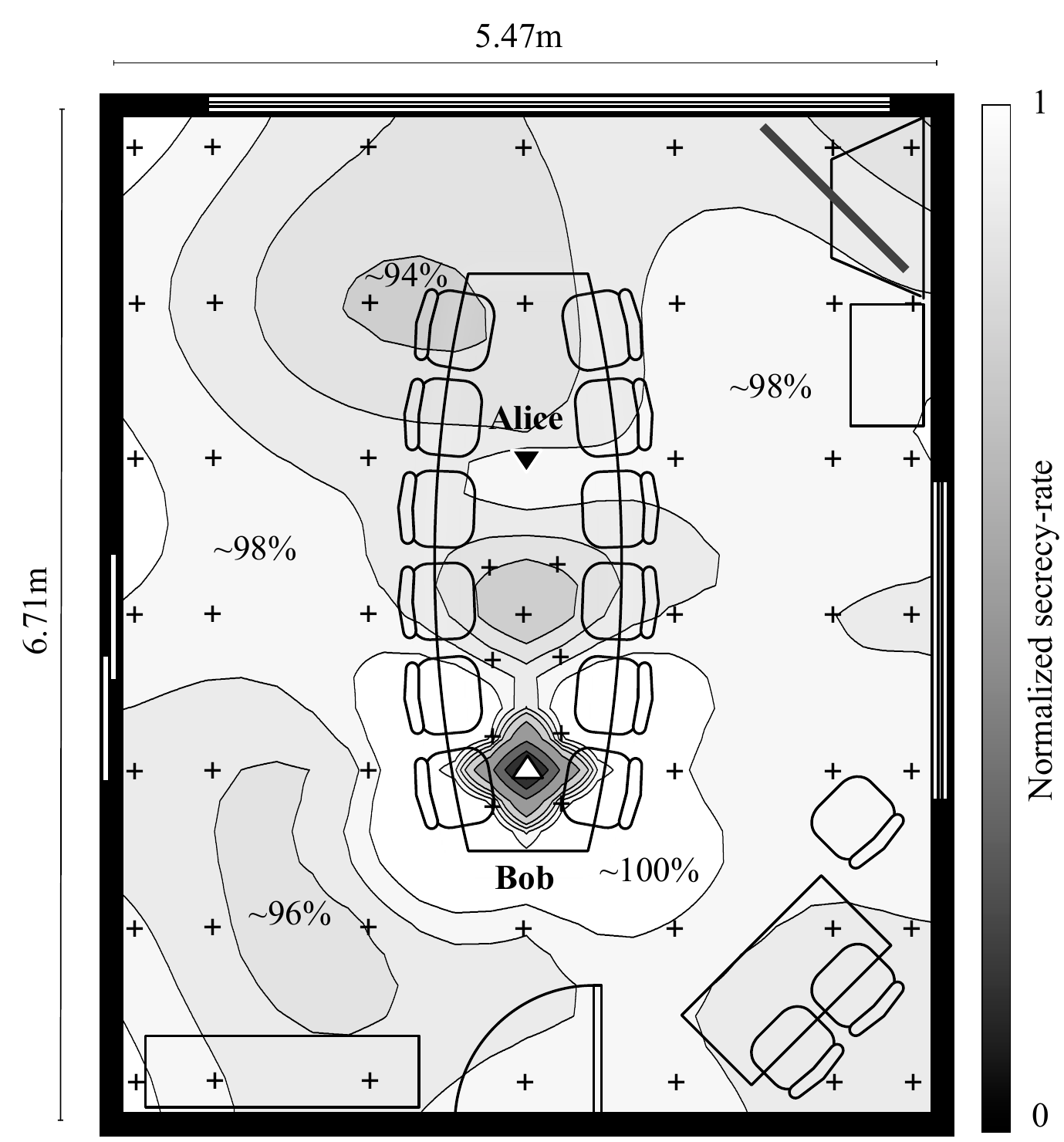}}{
\includegraphics[width=11.5cm]{nomotion.pdf}
}
\caption{Measurements of the normalized secrecy-rate  without motion}
\label{fig:nomotion}
\end{figure}

\begin{figure}
\centering
    \ifthenelse{\boolean{double}}{
\includegraphics[width=8.5cm]{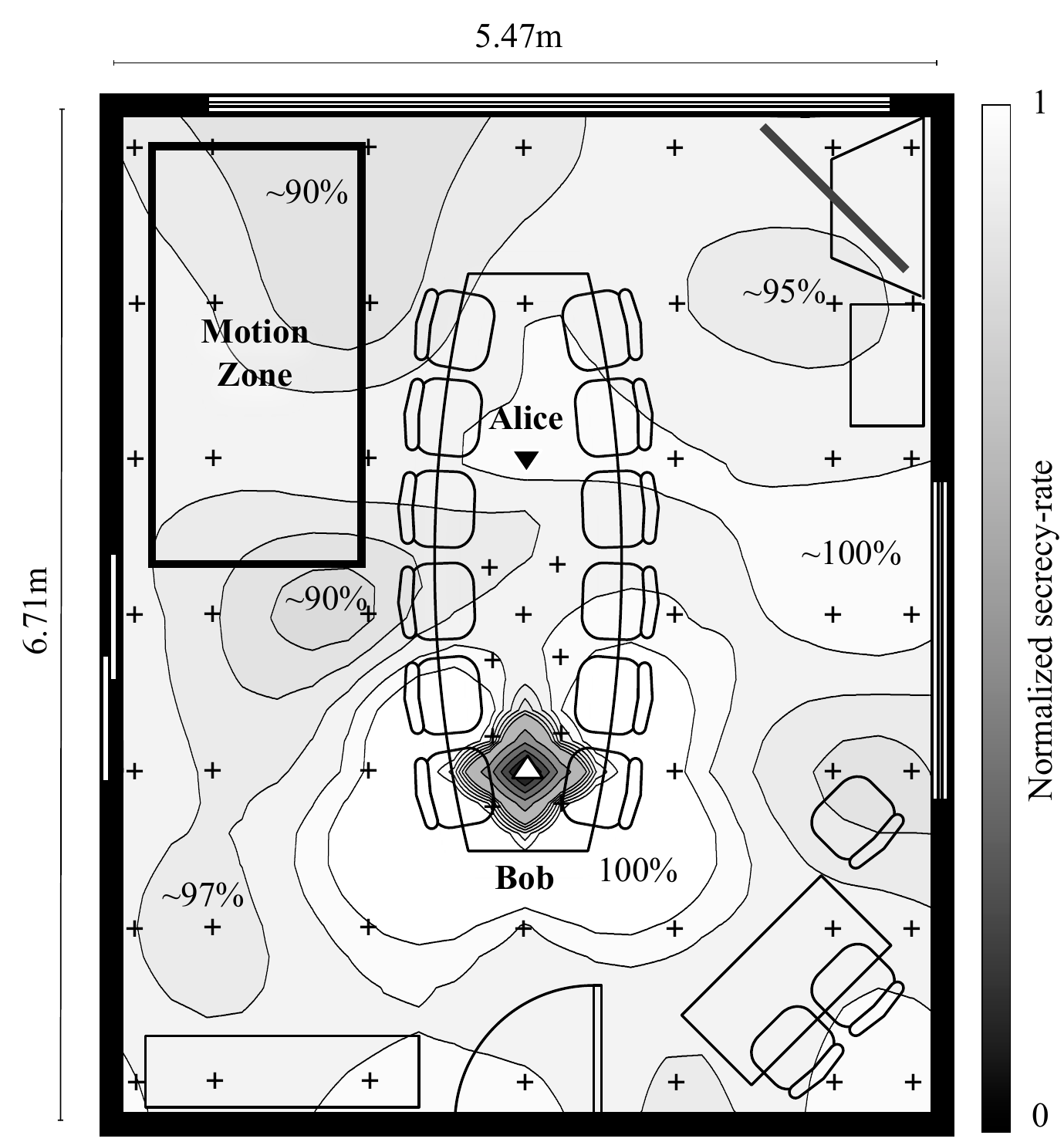}}{
\includegraphics[width=11.5cm]{withmotion.pdf}
}
\caption{Measurements of the normalized secrecy-rate  with motion}
\label{fig:withmotion}
\end{figure}

The first series of measurement shown in Fig.~\ref{fig:nomotion} is conducted without movement and serves as a benchmark. In this situation, there is no fluctuation of the channel gains, except those induced by the noise at the receivers' terminals. Therefore, the quantity $\avgI{G_{AB};G_{AE}}$ is small since the receiver noise is independent from one radio to another. When Eve and Alice use the same antenna, we obtain a darker spot since we create an electronic coupling between the terminals. 

From a secret-key generation standpoint, we need to introduce motion to induce variations of the channel gains. In a second series of measurement shown in Fig.~\ref{fig:withmotion}, the operator is walking in the upper left corner of the room. We observed high correlations when Bob and Eve's antennas are huddled together, and a fast decay of this correlation with distance, with leaked information reaching almost zero after a few centimeters. However, we observe that the leaked information increased again further away, even reaching values as high as 10\% in the upper left corner. Since this corner was actually the place where the operator was moving during the experiment, it suggests that measurements close to the motion source provide a better insight into the legitimate channel fluctuations. Therefore, defining a simple zone of guaranteed secrecy for key generation is not straightforward. From a security standpoint this clearly shows that we cannot ignore the information leaked to the eavesdropper when channel variations come from the motion in the environment. It shows that a secret key generation \emph{must} include a stage of privacy amplification to deal with unforeseeable levels of leaked information.

\section{Experimental Secret-Key Generation in the Finite Block-Length Regime}
\label{sec:finlen}


%
Once the statistics $p_{X_QY_QZ}$ of the source are characterized, one can easily calculate asymptotic achievable secret-key rates $\avgI{X_Q;Y_Q}-\avgI{X_Q;Z}$. However, these rates may be far too optimistic when operating on a finite number of samples, and it is crucial to avoid overestimating the number of secret bits that one can effectively extract with reconciliation and privacy amplification. Our analysis in Section~\ref{sec:finite-length-analys} is based on the detailed study of privacy amplification with continuos eavesdropper's observation, which differs from the finite-length analysis in~\cite{Watanabe2013,Tyagi2014arXiv} restricted to discrete observations. Our numerical results in Section~\ref{sec:numerical-results} are also obtained for the memoryless source $p_{X_QY_QZ}$ characterized experimentally in Section~\ref{sec:esi}, and not from computer simulations. We emphasize that the only approximation in our analysis is that the source statistics estimated in Section~\ref{ssec:char} correspond to the true statistics; the entire analysis in Section~\ref{sec:finite-length-analys} is exact.

\subsection{Finite-length analysis of privacy amplification for a continuous observation $Z$}
\label{sec:finite-length-analys}

We now analyze a sequential strategy~\cite{Maurer00,Bloch11}, in which the reconciliation step is performed with error correction codes, such as~\cite{Bloch06}, and the privacy amplification step is performed with universal$_2$ hash functions~\cite{Bennett95}. The major technical challenge is that $Z$ is continuous,  so that it is not possible to directly use previous approaches~\cite{Maurer00,Bennett95}, which are only valid for discrete random variables. In principle, we could quantize $Z$, since by \cite{Pinsker64}\cite{Fano61}\cite[Section 8.5]{Cover91}\cite[Lemma 2]{Barros08}\cite[Lemma 1.2]{Chou12b}, for any $\delta >0$, if a quantized version $Z_{Q'}$ of $Z$ is fine enough, we have
\begin{equation*}
|\mathbb{I}(K;AZ) -\mathbb{I}(K;FZ_{Q'})| < \delta.
\end{equation*}
Unfortunately, this result is only useful asymptotically since we do not know how to relate a fixed quantized version $Z_{Q'}$ to $\delta$. We circumvent the problem using the notion of \emph{continuous strong typicality}~\cite{Koetter11}, and we refer the reader to Appendix~\ref{sec:conttyp} for all notation and definitions related to continuous strong typicality. 

We assume that reconciliation is performed on $X_Q$ and $Y_Q$, and we define $\mathbf{P}_e^{\textup{rec}}$ as the probability of error of the reconciliation step, and $l_{\textup{rec}}$ as the number of information bits leaked during the process. To determine the final secret-key length obtained after privacy amplification with a universal$_2$ hash function $G$ chosen at random, we lower bound $\mathbb{H}(K |G Z^{n}F)$, which represents the uncertainty on the key the eavesdropper gets with its own observation $Z^n$,  the knowledge of $G$, and the public message $F$. We first lower bound the R\'{e}́nyi entropy $\mathbb{H}_{c}(X_Q^n|Z^n=z^n,F=f_{rec})$~\cite{Bennett95} by a term that can be numerically evaluated, and then apply~\cite[Theorem 4.4]{Bloch11}. Our analysis relies on two technical lemmas. Lemma~\ref{lem1} shows that the R\'{e}́nyi entropy $\mathbb{H}_{c}(X_Q^n|Z^n=z^n)$ is lower-bounded by the Shannon entropy $n\mathbb{H}(X_Q|Z)$ and a carefully characterized term that vanishes as $n$ goes to infinity. Lemma~\ref{lem2} relates $\mathbb{H}_{c}(X_Q^n|Z^n=z^n)$ to $\mathbb{H}_{c}(X_Q^n|Z^n=z^n,F=f_{rec})$.

\begin{lemma}[Adapted from \cite{Maurer00}\cite{Bloch11}] \label{lem1}

Consider a memoryless source $(\mathcal{X}_Q\mathcal{Z},p_{XZ})$ with $\mathcal{X}_Q$ a discrete alphabet and $\mathcal{Z}$ a continuous alphabet. Consider 
$\mathcal{A}_{\epsilon,\epsilon'}^{(n)}(Z,X_Q)$, as defined in Appendix~\ref{sec:conttyp}, where $\epsilon'$ is fixed and chosen such that
\begin{equation} \label{epscond}
  \exists n_0 \in \mathbb{N}, \forall n> n_0,  \sum_{i=1}^2 \left( e^{ n C_{Z}^{(i)}(\epsilon') } + e^{ n C_{ZX_Q}^{(i)}(\epsilon') } \right) \leq 2^{-b_1 n \epsilon}.
\end{equation}

Define 
$$\Theta \triangleq \mathds{1}\left\{(Z^n,X_Q^n) \in  \mathcal{A}_{\epsilon,\epsilon'}^{(n)}(Z,X_Q) \text{ and } Z^n \in \mathcal{B}_{\epsilon,\epsilon'}^{(n)}(Z,X_Q)\right\},   $$
where  
\begin{multline*}
\mathcal{B}_{\epsilon,\epsilon'}^{(n)}(Z,X_Q) \triangleq  \\ \left\{ z^n: \mathbb{P} \left[ \left( \mathcal{A}_{\epsilon,\epsilon'}^{(n)} (Z,X_Q) \right)^{\mathsf c} \Big| Z^n =z^n  \right] \leq 2^{-n b_0 \epsilon} \right\}.
\end{multline*}
 Then, for $n>n_0$, 
$$\mathbb{P}[\Theta=1]\geq 1-\delta_{\epsilon}^{(1)}(n),$$
 where $\delta_{\epsilon}^{(1)}(n) \triangleq  \sum_{i=1}^2 e^{nC_Z^{(i)}(\epsilon)} + 2^{-nb_1\epsilon} + 2^{-n(b_1-b_0)\epsilon}$. Moreover, for $z^n \in \mathcal{B}_{\epsilon,\epsilon'}^{(n)}(Z,X_Q)$,
\begin{align*}
 \mathbb{H}_{c}(X_Q^n|Z^n=z^n,\Theta=1)  \geq \ifthenelse{\boolean{double}}{\\}{} n(\mathbb{H}(X_Q|Z) -\epsilon - \epsilon') + \log_2 ( 1- \delta_{\epsilon}^{(2)}(n)),
\end{align*}
with $\delta_{\epsilon}^{(2)}(n) \triangleq  2^{-nb_0\epsilon}$.
\end{lemma}

\begin{proof}
See Appendix \ref{proof_lem1}
\end{proof}

\begin{lemma}[\!\cite{Maurer00,Bloch11}] \label{lem2}
Let $S \in \mathcal{S}$ and $U \in \mathcal{U}$ be two discrete random variables with joint distribution $p_{SU}$. Let $f \in \mathbb{R}^{\mathbb{N}}$. Define the function  $\Upsilon : \mathcal{U} \mapsto \left\{ 0,1\right\} $ as $$
\Upsilon(u)  \triangleq \mathds{1}\left\{ \mathbb{H}_{c}(S) -  \mathbb{H}_{c}(S|u) \leq \log|\mathcal{U}|+ 2 f(n) + 2\right\}.$$
Then, $\mathbb{P}_U[\Upsilon(U)=1]\geq1-2^{-f(n)}$.
\end{lemma}

We are now ready to develop the finite length analysis. Let $b_0 >0$, $b_1>b_0$, $\epsilon>0$, and $f(n),g(n) \in \mathbb{R}^{\mathbb{N}}$. 
We define 
\begin{align*}
\Theta  &\triangleq  \mathds{1}\left\{( Z^n, X_Q^n) \in  \mathcal{A}_{\epsilon,\epsilon'}^{(n)}(Z,X_Q) \text{ and } Z^n \in \mathcal{B}_{\epsilon,\epsilon'}^{(n)}(Z,X_Q) \right\},\\
\Upsilon  &\triangleq \mathds{1}\big\{\mathbb{H}_{c}(X_Q^{n} | Z^n = z^{n}, \Theta =1) \ifthenelse{\boolean{double}}{\\ &\hspace{2cm}}{}- \mathbb{H}_{c}(X^{n} | Z^n = z^{n}, \Theta =1, F=f_{rec}) \ifthenelse{\boolean{double}}{\\ & \hspace*{5cm}}{}\leq \log|\mathcal{F}| +2 f(n) +2\big\}.
\end{align*}
By Lemmas \ref{lem1} and \ref{lem2}, $\mathbb{P}(\Upsilon=1,\Theta=1) \geq 1 - \delta_{\epsilon,f}^{(3)}(n)$, with $\delta_{\epsilon,f}^{(3)}(n)  \triangleq 2^{-f(n)} + \delta_{\epsilon}^{(1)}(n) $
and \vspace*{-0.cm} \begin{align}
\noindent\mathbb{H}(K|GZ^{n}F)   
 \geq \left(1 - \delta_{\epsilon,f}^{(3)}(n) \right) \mathbb{H}(K |G Z^{n}F, \Upsilon = 1, \Theta=1).
  \label{eqk0}
\end{align}
To lower bound $\mathbb{H}(K |G Z^{n}F, \Upsilon=1, \Theta=1)$, we lower bound $\mathbb{H}_{c}(X^{N} | Z^{n}=z^{n}, F=f_{rec}, \Theta =1, \Upsilon =1)$ for any $z^n \in \mathcal{B}_{\epsilon,\epsilon'}^{(n)}(Z,X_Q)$ by means of Lemmas \ref{lem1} and \ref{lem2}, and use \cite[Theorem 4.4]{Bloch11} \cite{Bennett95}. By definition of $\Upsilon$,\vspace*{-0.cm}
\begin{multline}
  \mathbb{H}_{c}(X_Q^{n} | Z^n = z^{n},F=f_{rec}, \Theta =1, \Upsilon =1)   \ifthenelse{\boolean{double}}{\\}{}
\stackrel{\phantom{}}{\geq} \mathbb{H}_{c}(X_Q^{n} | Z^{n}=z^{n}, \Theta =1) - \log |\mathcal{F}| - 2f(n) -2 \ifthenelse{\boolean{double}}{\\}{}
\stackrel{(\ast)}{\geq}  \tilde{k},\label{eqdeb2}
\end{multline}
where $(\ast)$ follows from Lemma \ref{lem1}, with
$$ \tilde{k} \triangleq n(H(X_Q|Z) -\epsilon - \epsilon') + \log_2 ( 1- \delta_{\epsilon}^{(2)}(n)) - l_{\textup{rec}} - 2f(n) -2. 
$$
We set $\bar{k} \triangleq \left\lfloor \tilde{k} -g(n) \right\rfloor $, and, by \cite[Theorem 4.4]{Bloch11} \cite{Bennett95} and (\ref{eqdeb2}), we obtain
\begin{equation}
\mathbb{H}(K |G, Z^{n}=z^n,F=f_{rec}, \Upsilon=1, \Theta=1) \geq \bar{k} -\delta_g^{(4)}(n), \label{eqk1}
\end{equation}
with $\delta_g^{(4)}(n) \triangleq \log \left( 1+ 2^{-g(n)} \right)$.
Then, by (\ref{eqk0}) and (\ref{eqk1})
$$\mathbb{H}(K|GZ^{n}F) \geq \left(1 - \delta_{\epsilon,f}^{(3)}(n)\right) \left(\bar{k}- \delta_g^{(4)}(n) \right).  
$$
To summarize, we obtain
\begin{align*}
\mathbf{P}_e  \leq \mathbf{P}_e^{\textup{rec}}, \text{ }
\mathbf{U}_e  \leq \delta_{\epsilon,f,g}^{(5)}(n), \text{ }
\mathbf{L}_e  \leq  \delta_{\epsilon,f,g}^{(5)}(n),
\end{align*}
\begin{multline*}
\bar{k} (\epsilon,\epsilon',b_0,b_1,f,g) \ifthenelse{\boolean{double}}{\\}{} \triangleq \bigg\lfloor n(\mathbb{H}(X_Q|Z) -\epsilon - \epsilon') + \log_2 ( 1- \delta_{\epsilon}^{(2)}(n)) \ifthenelse{\boolean{double}}{\\}{}- l_{\textup{rec}} - 2f(n) -g(n) -2 \bigg\rfloor,
\end{multline*}
with
\vspace*{-0em}
\begin{align*}
\delta_{\epsilon}^{(1)}(n) & \triangleq \sum_{i=1}^2 e^{ nC_Z^{(i)}(\epsilon) } + 2^{-nb_1\epsilon} + 2^{-n(b_1-b_0)\epsilon},\\
\delta_{\epsilon}^{(2)}(n) & \triangleq  2^{-nb_0\epsilon},\\
\delta_{\epsilon,f}^{(3)}(n)  & \triangleq 2^{-f(n)} + \delta_{\epsilon}^{(1)}(n), \\
\delta^{(4)}_g(n) & \triangleq  \log_2 \left(1+2^{-g(n)} \right), \\
\delta_{\epsilon,f,g}^{(5)}(n) & \triangleq \delta_{\epsilon,f}^{(3)}(n)  \left( k -\delta^{(4)}_g(n) \right) + \delta^{(4)}_g(n).
\end{align*}
%

Recall that $(\mathcal{X}_Q\mathcal{Z},p_{X_QZ})$ is a memoryless source with ${X}_Q$ discrete and ${Z}$ continuous with known statistics. 
%
For fixed $\epsilon_L>0$, $\epsilon_U>0$, and block-length $n$, sequential secret-key generation with privacy amplification performed with universal hash function therefore ensures $\mathbf{L}<\epsilon_L$, $\mathbf{U}< \epsilon_U$ and $\mathbf{P}_{\mathrm{d}}  \leq \mathbf{P}_{\textup{rec}}$, while achieving the following generated secret-key length 

\begin{equation} \label{eq:kbound}
k \triangleq  \sup_{  (\epsilon,\epsilon',b_0,b_1,f,g ) \in \mathcal{C}(\epsilon_L,\epsilon_U)} \bar{k} (\epsilon,\epsilon',b_0,b_1,f,g),
\end{equation}
where
\begin{multline*}
\mathcal{C}(\epsilon_L,\epsilon_U) \triangleq \Big\{ (\epsilon,\epsilon',b_0,b_1,f,g) \in [0,1]^4 \times\mathbb{R}_+^{\mathbb{N}} \times \mathbb{R}_+^{\mathbb{N}} : \\ 
 \exists n_0, \forall n> n_0,  \sum_{i=1}^2 \left( e^{nC_{Z}^{(i)}(\epsilon')} + e^{nC_{ZX_Q}^{(i)}(\epsilon')} \right) \leq 2^{-b_1 n \epsilon}, \\ b_1 >b_0,\, f(n)=o(n),\, g(n)=o(n),\,  \delta_{\epsilon,f,g}^{(5)}(n) \leq \min(\epsilon_L,\epsilon_U)  \Big\}.
\end{multline*}
%

%
%
%
%
Note that, asymptotically, the corresponding achievable key rate is $ R_{\textup{low}} \triangleq \mathbb{I}(X_Q;Y_Q) - \mathbb{I}(X_Q;Z)$, which is a lower bound of the secret-key capacity $C_s$ \cite{Maurer93}. 
\subsection{Numerical results}
\label{sec:numerical-results}
 We now use the experimental measurements in Section~\ref{sec:esi} to evaluate our achievable bound in~\eqref{eq:kbound}. Without losing generality, and to separate the finite length effect of reconciliation and privacy amplification, we also assume that the reconciliation has efficiency $\beta \in [0,1]$ (see \cite{Bloch11})  so that $l_{\textup{rec}} = n (\mathbb{H}(X) - \beta {I}(X;Y))$ bits are leaked to the eavesdropper. The best case scenario, obtained for $\beta =1$, would yield $\mathbb{H}(X|Y)$ bits leaked during the reconciliation step. Next, to compute the key length, we estimated the quantities 
$$\mathbb{H}(X_{Q}|Z), \text{ } C_{Z}^{(i)}(\epsilon),  \text{ } C_{ZX_Q}^{(i)}(\epsilon),$$
for some $\epsilon>0$ and $i\in \llbracket 1,2 \rrbracket$. We then numerically optimized an estimate of $k$ in~\eqref{eq:kbound} for which $\mathbf{U}<10^{-3}$, $\mathbf{L}<10^{-3}$, by testing a large range of parameter values $(\epsilon,\epsilon',b_0,b_1,f,g)$, where $f$ and $g$ are taken of the form $f(n) \triangleq n^{-\alpha_1}$, $g(n) \triangleq n^{-\alpha_2}$, with $n \in \mathbb{N}$, $\alpha_1 >0$ and $\alpha_2>0$. Note that, to ensure the first constraint in $\mathcal{C}(\epsilon_L,\epsilon_U)$, we guarantee instead the sufficient condition $4 \exp \left[ n  C_{\textup{max}} \right] \leq 2^{-b_1n\epsilon}$, i.e., $  n (C_{\textup{max}}+b_1 \epsilon \ln 2) + 2 \ln 2 \leq 0 $, where $C_{\textup{max}} \triangleq \max_{i\in\llbracket 1,2\rrbracket} \left(C_{Z}^{(i)}(\epsilon'),C_{ZX_Q}^{(i)}(\epsilon') \right)$.  The key rates are finally reported in Fig.~\ref{fig:ratio} in terms of the ratio $\eta$ of the finite-length key rate to the asymptotic key-rate defined as
$$
\eta \triangleq \frac{k/n}{R_{\textup{low}}},
$$
where $R_{s}= \avgI{X_Q;Y_Q}-\avgI{X_Q;Z}$. For our specific statistics $p_{X_QY_QZ}$, we have $R_{s}\approx 1.46 \text{ bits}$.\footnote{Note that $R_s$ is only an achievable key rate. The secret-key capacity $C_s$ is only know to satisfy $C_s \leq  \mathbb{I}(X;Y|Z) \approx 1.76 \text{ bits}$.} The corresponding values of the parameters $(\epsilon,\epsilon',b_0,b_1,f,g)$ are also reported in Table~\ref{table}.

%
%


\begin{figure}
\centering 
    \ifthenelse{\boolean{double}}{
    \setlength\figureheight{6cm} 
    \setlength\figurewidth{7cm}}{
    \setlength\figureheight{7cm} 
    \setlength\figurewidth{10cm}} 
%
%
%
\definecolor{mycolor1}{rgb}{0.1894645708138,0.1894645708138,0.1894645708138}%
\definecolor{mycolor2}{rgb}{0.287174588749259,0.287174588749259,0.287174588749259}%
\definecolor{mycolor3}{rgb}{0.435275281648062,0.435275281648062,0.435275281648062}%
\definecolor{mycolor4}{rgb}{0.659753955386447,0.659753955386447,0.659753955386447}%
\begin{tikzpicture}

\begin{axis}[%
width=\figurewidth,
height=\figureheight,
scale only axis,
xmin=0,
xmax=50000,
xtick={    0, 10000, 20000, 30000, 40000, 50000},
xlabel={Block-length $n$},
xmajorgrids,
ymin=0.3,
ymax=1,
ytick={0.3, 0.4, 0.5, 0.6, 0.7, 0.8, 0.9,   1},
ylabel={$\text{Ratio }\eta$},
ymajorgrids,
axis x line*=bottom,
axis y line*=left,
legend style={at={(0.628754208754208,0.0793893129771)},anchor=south west,draw=black,fill=white,legend cell align=left}
]
\addplot [
color=darkgray!50!black,
solid,
smooth,
very thick,
mark=*,
mark options={solid,fill=black,draw=black}
]
table[row sep=crcr]{
1000 0.573817300632861\\
2000 0.709954264009578\\
5000 0.819952930417965\\
10000 0.871276565610988\\
20000 0.910245771377573\\
30000 0.924959908169206\\
50000 0.943143268577557\\
};
\addlegendentry{$\beta = 1$};

\addplot [
color=mycolor1,
solid,
smooth,
very thick,
mark=*,
mark options={solid}
]
table[row sep=crcr]{
1000 0.51203422290661\\
2000 0.648171186283327\\
5000 0.758169852691714\\
10000 0.809493487884736\\
20000 0.848462693651321\\
30000 0.863176830442955\\
50000 0.881360190851305\\
};
\addlegendentry{$\beta = 0.95$};

\addplot [
color=mycolor2,
solid,
smooth,
very thick,
mark=*,
mark options={solid}
]
table[row sep=crcr]{
1000 0.450251145180359\\
2000 0.586388108557076\\
5000 0.696386774965463\\
10000 0.747710410158485\\
20000 0.78667961592507\\
30000 0.801393752716704\\
50000 0.819577113125054\\
};
\addlegendentry{$\beta = 0.90$};

\addplot [
color=mycolor3,
solid,
smooth,
very thick,
mark=*,
mark options={solid}
]
table[row sep=crcr]{
1000 0.388468067454107\\
2000 0.524605030830824\\
5000 0.634603697239211\\
10000 0.685927332432234\\
20000 0.724896538198819\\
30000 0.739610674990452\\
50000 0.757794035398802\\
};
\addlegendentry{$\beta = 0.85$};

\addplot [
color=mycolor4,
solid,
smooth,
very thick,
mark=*,
mark options={solid}
]
table[row sep=crcr]{
1000 0.326684989727856\\
2000 0.462821953104573\\
5000 0.57282061951296\\
10000 0.624144254705982\\
20000 0.663113460472567\\
30000 0.677827597264201\\
50000 0.696010957672551\\
};
\addlegendentry{$\beta = 0.80$};

\end{axis}
\end{tikzpicture}%
\caption{Ratio $\eta$, for $\mathbf{U}<10^{-3}$ and $\mathbf{L}<10^{-3}$}
\label{fig:ratio}
\end{figure}
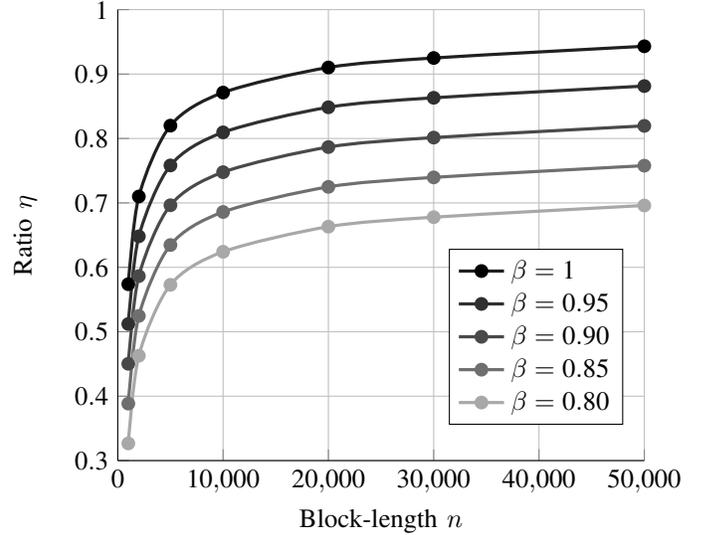 

Fig.~\ref{fig:ratio} shows that the achievable secret-key rates may be much lower than the upper bound $R_{\textup{low}}$ in the finite block-length regime. For instance, even if the reconciliation protocol has an efficiency of 90\% for $n=$30,000, the actual achievable secret-key rate is 20\% lower than $R_{\textup{low}}$ as shown in Fig.~\ref{fig:ratio}. 

\begin{table} \footnotesize
\renewcommand\arraystretch{1.4}  
\begin{center}
 \caption{\label{table} Parameters for the values of $\eta$ on Figure \ref{fig:ratio}}

\begin{tabular}{|c|} 
 \hline
 $n$  \\ \hline
 \hline
 $\epsilon$ \\
 \hline
 $\epsilon'$ \\
 \hline
 $b_1$ \\
 \hline
 $b_0$  \\
 \hline
 $\alpha_1$ \\
 \hline
 $\alpha_2$  \\
 \hline
\end{tabular}\hspace*{-1pt}
\begin{tabular}{|c|c|c|c|c|c|} 
 \hline
  1,000  & 2,000 & 5,000 & 10,000 & 20,000 & 50,000 \\ \hline
 \hline
 0.2590 & 0.1840 & 0.1199 & 0.0900 & 0.0625 & 0.0410\\
 \hline
 0.3019 & 0.2034 & 0.1286 & 0.0912 & 0.0653 &  0.0407\\
 \hline
 0.1000 & 0.0660 &  0.0420 & 0.0289 & 0.0220 & 0.0140\\
 \hline
0.0100 & 0.0010 &  0.0010 & 0.0010 & 0.0010 & 0.0010 \\
 \hline
 0.4400 & 0.4200 & 0.3800 & 0.3500 & 0.3300 & 0.3100 \\
 \hline
 0.4400 & 0.4200 & 0.3800 & 0.3500 & 0.3300 & 0.3100 \\
 \hline 
\end{tabular}
\end{center}
\end{table}

\section{Conclusion and discussion}
\label{sec:concl-disc}


The main weakness of secret-key generation from wireless channel gains is the difficulty to clearly establish conditions for secrecy without a precise knowledge of the eavesdropper's statistics. Our experiments show that there exists no simple relationship between the eavesdropper's proximity and the correlation of its observations in a typical indoor environment, which suggests that traditional parametric models of wireless channels should be used with great care. Furthermore, even if characterizing the wireless environment is possible, a precise estimation of the source statistics is a challenge in itself, which has a non-negligible effect on the estimation of actual achievable key-rates and the operation of privacy amplification with a finite number of samples.

To address this weakness, several solutions are worth investigating. First, one might want to restrict the use of secret-key generation schemes to situations in which the channel gains possess more entropy, and in which an eavesdropper's measurements would exhibit less statistical dependence. As suggested by the results in~\cite{Mathur2008,Ye2010}, high-mobility environment might be more suitable, and one might consider exploiting the phase of the complex channel gain instead of the magnitude. However, exploiting the phase is technically more challenging, and would require precise synchronization between terminals. Second, one might want to turn to \emph{channel models} for secret-key generation~\cite{Ahlswede1993}, in which a legitimate terminal injects artificial noise in the channel, combined with multiple antennas. In particular, the results of~\cite{He2010, Chou13c} suggest that one can get away with a mere assumption regarding the number of antennas of the eavesdropper, which completely removes the need for the intricate eavesdropper's statistics estimation.

\appendices

\section{Continuous Typicality} \label{sec:conttyp}
										
This section is adapted from~\cite{Koetter11}.

Let $\mathcal{X}$, $\mathcal{Y}$ be two discrete or continuous alphabets. Let $b_0 >0$, $b_1>b_0$ and $\epsilon>0$. Let $n \in \mathbb{N}$. Let $\{ (X_i,Y_i)\}_{i=1}^n$ be i.i.d. random variables drawn according to the joint distribution $p_{XY}$.  Define
$$
f(x^n) \triangleq \left| - \frac{1}{n} \log p(x^n) - H(X) \right|,
$$
$$
f(y^n) \triangleq \left| - \frac{1}{n} \log p(y^n) - H(Y) \right|,
$$
$$
f(x^n,y^n) \triangleq \left| - \frac{1}{n} \log p(x^n,y^n) - H(X,Y) \right|,
$$
where $H$ denotes either the discrete entropy $\mathbb{H}$ or the differential entropy $\mathds{h}$.
Define
\begin{multline*}
\mathcal{A}_{\epsilon,\epsilon'}^{(n)}(X,Y) \triangleq \big\{  ({x}^n, {y}^n) \in \mathcal{X}^n \times \mathcal{Y}^n  : \ifthenelse{\boolean{double}}{\\}{} f(x^n) \leq \epsilon , \text{ }f(y^n) \leq \epsilon', \text{ }f(x^n, y^n) \leq \epsilon' \big\},
\end{multline*}
where $\epsilon'$ is such that\footnote{
Because we want to numerically compute $\epsilon'$, we do not do the better choice 
\begin{multline*}
\epsilon' \triangleq (1+ \epsilon) \inf \big\{ \epsilon'' >0 : \ifthenelse{\boolean{double}}{\\}{}\exists n_0, \forall n> n_0, \mathbb{P} \left[ f(Y^n) > \epsilon'' \text{ or } f(X^n,Y^n)  > \epsilon'' \right]  \leq 2^{-b_1N \epsilon}  \big\}, 
\end{multline*}
as in \cite{Koetter11}.
}
$$
\exists n_0, \forall n> n_0, \mathbb{P} \left[ f(Y^n) > \epsilon'\text{ or } f(X^n,Y^n)  > \epsilon'\right]  \leq 2^{-nb_1  \epsilon}. 
 $$
\begin{remark}
Note that $\epsilon'$ is a function of $\epsilon$ and is not uniquely defined for a fixed $\epsilon$.
\end{remark}
Define 
\begin{multline*}
\widehat{\mathcal{A}}_{\epsilon,\epsilon'}^{(n)}(X,Y) \triangleq \Bigg\{  (x^n, y^n) \in \mathcal{A}_{\epsilon,\epsilon'}^{(n)}(X,Y) : \ifthenelse{\boolean{double}}{\\}{}\mathbb{P} \left[  \left( \mathcal{A}_{\epsilon,\epsilon'}^{(n)}(X,Y) \right)^{\mathsf c}  \bigg| X^n = x^n\right] \leq 2^{-nb_0\epsilon} \Bigg\},
\end{multline*}
with 
$$
\mathbb{P} \left[  \left( \mathcal{A}_{\epsilon,\epsilon'}^{(n)}(X,Y) \right)^{\mathsf c}  \bigg| X^n = x^n \right] \triangleq \sum_{y^n:(x^n,y^n) \notin \mathcal{A}_{\epsilon,\epsilon'}^{(n)}} p(y^n|x^n).
$$
Define
$$
C_X^{(1)}(\epsilon) \triangleq  \min_{s>0} \left( \ln \mathbb{E} [p(X)^{-s/ \ln 2}] - s(H(X)+ \epsilon) \right),
$$ 
$$
C_X^{(2)}(\epsilon) \triangleq \min_{s>0} \left( \ln \mathbb{E} [p(X)^{s / \ln 2}] - s( - H(X)+ \epsilon) \right) .
$$
We similarly define $C_Y^{(i)}(\epsilon)$ and $C_{XY}^{(i)}(\epsilon)$ for $i \in \llbracket 1,2 \rrbracket$.
\begin{remark}
For $i \in \llbracket 1,2 \rrbracket$, the expression $C_X^{(i)}(\epsilon)$, $C_Y^{(i)}(\epsilon)$ and $C_{XY}^{(i)}(\epsilon)$ can be specified. For instance, in $C_X^{(1)}(\epsilon)$, 
$$
s_0 \triangleq \mathrm{argmin}_{s>0} \left( \ln \mathbb{E} [p(X)^{-s/ \ln 2}] - s (H(X) + \epsilon) \right),
$$
  is such that 
$$
\frac{1}{\ln 2}\frac { \mathbb{E}[-p(X)^{-s_0/ \ln 2} \log p(X)]}{\mathbb{E}[p(X)^{-s_0 / \ln 2 }] } = H(X) + \epsilon.
$$
\end{remark}
\begin{lemma}[Adapted from \cite{Koetter11}] \label{lem_cont}
Consider $\mathcal{A}_{\epsilon,\epsilon'}^{(n)}(X,Y)$, where $\epsilon'$ is fixed and chosen such that
$$
  \exists n_0, \forall n> n_0,  \sum_{i=1}^2 \left(  e^{ n C_Y^{(i)}(\epsilon') } + e^{ n C_{XY}^{(i)}(\epsilon')} \right) \leq 2^{-b_1 n\epsilon}.
$$
\begin{enumerate}
\item Let $(x^n,y^n) \in \mathcal{A}_{\epsilon,\epsilon'}^{(n)}(X,Y)$.
$$
p (y^n | x^n) \leq 2^{-n(H(Y|X)-\epsilon' - \epsilon)},
$$
\item Let $n>n_0$.
$$
\mathbb{P} \left[  \left( \mathcal{A}_{\epsilon,\epsilon'}^{(n)}(X,Y) \right)^{\mathsf c} \right]  \leq \sum_{i=1}^2 e^{nC_X^{(i)}(\epsilon)} + 2^{-nb_1\epsilon},
$$
\item  Let $n>n_0$.
$$
\mathbb{P} \left[  \left( \widehat{\mathcal{A}}_{\epsilon}^{(n)}(X,Y) \right)^{\mathsf c} \right] \leq \sum_{i=1}^2 e^{n C_X^{(i)}(\epsilon)} + 2^{-nb_1 \epsilon} + 2^{-n(b_1-b_0)\epsilon}.
$$
\end{enumerate}
\end{lemma}
\begin{proof} 
We first need to verify that $\epsilon'$ is well defined.

By Chernoff's bound, applied to $\{ -\log p(X_i)\}_{i=1}^n$, and $\{ \log p(X_i)\}_{i=1}^n$, we obtain
$$
\mathbb{P} \left[ -\frac{1}{N} \log p(X^n) > H(X) + \epsilon \right] \leq e^{nC_X^{(1)}(\epsilon)},
$$
and
$$
\mathbb{P} \left[ \frac{1}{N} \log p(X^n) > - H(X) + \epsilon \right] \leq e^{nC_X^{(2)}(\epsilon)}.
$$
Hence, by the union bound
$$
\mathbb{P} \left[ f(X^n) > \epsilon \right] \leq \sum_{i=1}^2 e^{nC_X^{(i)}(\epsilon)}. 
$$
Similarly,
$$
\mathbb{P} \left[ f(Y^n) > \epsilon \right] \leq \sum_{i=1}^2 e^{nC_Y^{(i)}(\epsilon)},
$$
$$
\mathbb{P} \left[ f(X^n,Y^n) > \epsilon \right] \leq \sum_{i=1}^2 e^{n C_{XY}^{(i)}(\epsilon)},
$$
Hence, $\epsilon'$ is well defined.
\begin{enumerate}
\item Let $(x^n,y^n) \in \mathcal{A}_{\epsilon,\epsilon'}^{(n)}(X,Y)$.

By definition of $\mathcal{A}_{\epsilon,\epsilon'}^{(n)}(X,Y)$, we have 
$$
p(x^n) \geq 2^{-n(H(X) + \epsilon)},
$$
$$
p(x^n,y^n) \leq 2^{ -n( H(XY)-\epsilon')},
$$
hence
$$
p(y^n|x^n) = \frac{p(x^n,y^n)}{p(x^n)} \leq 2^{-n( H(Y|X) - \epsilon' - \epsilon )}.
$$
\item Let $n>n_0$. By the union bound and the definition of $\epsilon'$,
\begin{align*}
\mathbb{P} \left[  \left( \mathcal{A}_{\epsilon,\epsilon'}^{(n)}(X,Y) \right)^{\mathsf c} \right] 
&\leq   \mathbb{P} \left[ f(X^n) > \epsilon \right] \ifthenelse{\boolean{double}}{\\&}{}+ \mathbb{P} [ f(Y^n) > \epsilon'  \text{ or } f(X^n, Y^n) > \epsilon']\\
& \leq \sum_{i=1}^2 e^{n C_X^{(i)}(\epsilon)} + 2^{-nb_1\epsilon}.
\end{align*}
\item Let $n>n_0$. 
Define
\begin{multline*}
\mathcal{C}^n \triangleq \Bigg\{ x^n \in \mathcal{X}^n : f(x^n) \leq \epsilon, \ifthenelse{\boolean{double}}{\\}{} \mathbb{P} \left[  \left( \mathcal{A}_{\epsilon,\epsilon'}^{(n)}(X,Y) \right)^{\mathsf c}  \bigg| X^n = x^n\right] > 2^{-nb_0\epsilon}  \Bigg\}, 
\end{multline*}
such that
\begin{align*}
&\mathbb{P} \left[  \left( \widehat{\mathcal{A}}_{\epsilon,\epsilon'}^{(n)}(X,Y) \right)^{\mathsf c} \right] - \mathbb{P} \left[  \left( \mathcal{A}_{\epsilon,\epsilon'}^{(n)}(X,Y) \right)^{\mathsf c} \right] \\
& = \mathbb{P} \left[ \left\{ (x^n, y^n) \in \mathcal{A}_{\epsilon,\epsilon'}^{(n)}(X,Y) : x^n \in \mathcal{C}^n \right\} \right] \\
& = \sum_{x^n \in \mathcal{C}^n} p(x^n) \mathbb{P} \left[ \mathcal{A}_{\epsilon,\epsilon'}^{(n)}(X,Y) \bigg| X^n = x^n \right] \\
& \leq \sum_{x^n \in \mathcal{C}^n} p(x^n)\\ 
& \stackrel{(a)}{\leq} 2^{nb_0 \epsilon} \sum_{x^n \in \mathcal{C}^n} p(x^n) \ifthenelse{\boolean{double}}{\times \\& \hspace{1em}}{}\mathbb{P} \left[ f(Y^n) > \epsilon'  \text{ or } f(X^n, Y^n) > \epsilon' \bigg| X^n = x^n \right] \\
& \leq 2^{nb_0 \epsilon} \sum_{x^n \in \mathcal{X}^n} p(x^n) \ifthenelse{\boolean{double}}{\times \\& \hspace{1em}}{}\mathbb{P} \left[ f(Y^n) > \epsilon'  \text{ or } f(X^n, Y^n) > \epsilon' \bigg| X^n = x^n \right] \\
& \stackrel{(b)}{\leq}  2^{-n(b_1-b_0)\epsilon},
\end{align*}
where (a) holds by definition of $\mathcal{C}^n$ and (b) holds by definition of $\epsilon'$.
\end{enumerate}
\end{proof}										
										
\section{Proof of Lemma \ref{lem1}}	\label{proof_lem1}

Let $z^n \in \mathcal{B}_{\epsilon,\epsilon'}^{(n)}(Z,X_Q)$. By Lemma \ref{lem_cont} in Appendix \ref{sec:conttyp} , we have $\mathbb{P}[\Theta=1]\geq \mathbb{P}\left[ \widehat{\mathcal{A}}_{\epsilon,\epsilon'}^{(n)}(Z,X_Q)\right] \geq 1-\delta_{\epsilon}^{(1)}(n).$  Then, by Bayes' rule,
\begin{multline}\label{eq1_a}
 \mathbb{P} (X_Q^n =x_Q^n|Z^n =z^n,\Theta=1) 
 \ifthenelse{\boolean{double}}{\\}{} = \frac{\mathbb{P} [ \Theta=1|X_Q^n =x_Q^n, Z^n =z^n ] \mathbb{P} (X_Q^n =x_Q^n|Z^n =z^n)}{\mathbb{P} [ \Theta=1| Z^n =z^n ]}.
\end{multline}
%
We have by definition of $\Theta$
\begin{align}
\mathbb{P} [\Theta =1|Z^n=z^n ]
& = \mathbb{P} \left(\mathcal{A}_{\epsilon,\epsilon'}^{(n)}(Z,X_Q)| Z^n = z^n\right)  \nonumber \\ \nonumber
& \geq 1- 2^{-nb_0\epsilon}  \\
& = 1 - \delta_{\epsilon}^{(2)}(n). \label{eq1_c}
\end{align}
We also have by Lemma \ref{lem_cont}.1
\begin{align} \label{eq1_b}
\mathbb{P}(X_Q^n =x_Q^n|Z^n =z^n) \leq 2^{ -n(\mathbb{H}(X_Q|Z) -\epsilon - \epsilon')},
\end{align}
for $(x_Q^n,z^n)$ satisfying $\mathbb{P} [ \Theta=1|X_Q^n =x_Q^n, Z^n =z^n ]>0$.\\
Hence, by (\ref{eq1_c}) and (\ref{eq1_b}), (\ref{eq1_a}) gives
\begin{align*}
\mathbb{P} (X_Q^n =x_Q^n|Z^n =z^n,\Theta=1) & \leq  \frac{2^{ -n(\mathbb{H}(X_Q|Z) -\epsilon - \epsilon')}}{1 - \delta_{\epsilon}^{(2)}(n)}. 
\end{align*}
Finally,
\begin{align*}
\mathbb{H}_{c}(X_Q^n\ifthenelse{\boolean{double}}{&}{}|Z^n=z^n,\Theta=1)  \ifthenelse{\boolean{double}}{\\}{}
& \geq \mathbb{H}_{\infty}(X_Q^n|Z^n=z^n,\Theta=1) \\
& = - \log_2 \max_{x_Q^n} \mathbb{P} (X_Q^n =x_Q^n|Z^n =z^n,\Theta=1) \\
& \geq n(\mathbb{H}(X_Q|Z) -\epsilon - \epsilon') + \log_2 ( 1- \delta_{\epsilon}^{(2)}(n)).
\end{align*}

\ifthenelse{\boolean{double}}{\IEEEtriggeratref{25}}{}
\bibliographystyle{IEEEtran} 
\bibliography{biblio}

\end{document}